\documentclass[10pt,conference]{IEEEtran}
\IEEEoverridecommandlockouts
\usepackage{cite}
\usepackage{amssymb,amsfonts}
\usepackage{algorithmic}
\usepackage{graphicx}
\usepackage{textcomp}
\usepackage{xcolor}
\usepackage{mathtools}
\usepackage{amsmath, amsthm}
\newtheorem{assumption}{Assumption}
\newtheorem{theorem}{Theorem}

\newtheorem{lemma}{Lemma}

\def\BibTeX{{\rm B\kern-.05em{\sc i\kern-.025em b}\kern-.08em
    T\kern-.1667em\lower.7ex\hbox{E}\kern-.125emX}}
\begin{document}

\title{Do the Defect Prediction Models Really Work?}

\author{\IEEEauthorblockN{Umamaheswara Sharma B}
\IEEEauthorblockA{Department of Computer Science and Engineering \\
National Institute of Technology, Warangal, India \\
uma.phd@student.nitw.ac.in}
\and
\IEEEauthorblockN{Ravichandra Sadam}
\IEEEauthorblockA{Department of Computer Science and Engineering\\
National Institute of Technology, Warangal, India \\
ravic@nitw.ac.in}
}

\maketitle

\begin{abstract}
``\textit{You may develop a potential prediction model, but how can I trust your model that it will benefit my software?}''. Using a software defect prediction (SDP) model as a tool, we address this fundamental problem in machine learning research. This is a preliminary work targeted at providing an analysis of the developed binary SDP model in real-time working environments.
\end{abstract}

\begin{IEEEkeywords}
Software Defect Prediction, Machine Learning, Probabilistic Bounds, Real-time Analysis.
\end{IEEEkeywords}

\section{Introduction}
\label{Introduction}
Due to the rapid development of complex and critical software systems, testing has become a tough challenge. Software defect prediction (SDP) models are being developed to alleviate this challenge \cite{khoshgoftaar1990predicting, lessmann2008benchmarking, herbold2017comparative, zhou2018far, zimmermann2009cross, amasaki2022extended}. The primary objectives in developing SDP models are to reduce the testing time, cost, and effort to be spent on the newly developed software project \cite{Sharma2022Measures}. The task of SDP models is to predict the defect proneness of newly developed software modules.

Once an efficient SDP model has been developed, any organisation may utilise its services. However, it is evident from the machine learning (ML) literature that, in general, the developed prediction models may produce misclassifications on the unseen data \cite{bishop2006pattern}. Owing to the result of either a misclassification (from the prediction model) or ineffective testing, the occurrence of any malfunction in the software modules may cause problems ranging from inconvenience to the loss of life \cite{lyu1996handbook}. It is more likely that a software fails when the prediction model wrongly predicts a defective module.

To know how feasible the SDP models are in the real-time working environments, we provide a theoretical analysis using probabilistic bounds. In a nutshell, the proofs are demonstrated by computing the deviation of a random variable (which is modeled as a hazard rate of a software that utilises SDP models) far from the estimated hazard rate of a manually tested software. Additionally, the proofs are also provided in terms of the measure called reliability.
\section{Preliminaries}
\label{Preliminaries}
We begin by discovering the chances of failures in the system from the predictions of SDP models. There are many ways a system will fail \cite{lyu1996handbook, pressman2005software}. Of which, the primary possible instance is when a defective module is predicted as clean. In such cases, in real-time working environments, the tester may miss the defective module. Now, the following assumption ensures failure incidents from each false negative module on the test set:
\begin{assumption}
\label{Assumption1}
Misclassification of each defective module can cause one failure in the software.
\end{assumption}

This assumption enables us to count the total failures on the test set and on the newly developed project. Since we know the fact that the general testing procedures do not prompt all the defects \cite{pressman2005software}, the following assumption ensures the presence of failures in any software that is tested by using SDP models:
\begin{assumption}
\label{Assumption2}
The integration test, system test, or acceptance test do not prompt the defects for the misclassified defective modules.
\end{assumption}

Now, to measure the percentage of occurrences of the failure cases on the test set, we use a measure called the false omission rate (FOR). The FOR is the ratio of the total number of false negatives over the total predicted clean modules. This is given as:
\begin{equation}
\label{FOR}
    \textit{FOR} = \frac{\text{False Negatives}}{\text{Predicted Cleans}} = \frac{\text{FN}}{\text{FN+TN}}
\end{equation}

Since only predicted clean modules may contain hidden defects, the measure FOR is well suited to estimating the percentage of failure occurrences on the test set. However, in real-time testing, FNs do not provide sufficient information about the failures in software because the actual class label for the predicted clean module is unknown. Hence, we model the actual class for each predicted clean module as a random variable. For any newly developed software $S = \{M_1, M_2, \cdots, M_n\}$ with $n$ modules, let us assume $l, (l \leq n)$ modules are predicted as being from the clean class. Now, the following random variable is used to represent the failure case from the wrongly predicted defective module, $M_i$:
\begin{equation}
\label{IndicatorRV}
X_i = \begin{dcases*}
        1, & if the module $M_i$ is classified wrongly as clean\\
        0, & if the module $M_i$ is classified correctly as clean
      \end{dcases*}
\end{equation}

To provide a guarantee that, for any module $M_i$, $X_i$ takes a value in $\{0,1\}$ with an identical probability, the following assumption must hold true.
\begin{assumption}
\label{Assumption3}
The SDP model is trained on the historical data of the software project(s).
\end{assumption}

In general, the SDP models are being developed on the historical data of the software projects, assuming similar data distributions for the training set, test set, and the population set \cite{khoshgoftaar1990predicting, lessmann2008benchmarking, herbold2017comparative, zhou2018far, zimmermann2009cross, amasaki2022extended, Sharma2022Measures}. From Assumption \ref{Assumption3}, since the SDP model does not change dynamically, we have that each predicted clean module goes into the wrong class with a similar FOR value. That is, the FOR is treated as the probability that each predicted clean module may fall into the defective module. This is given as:
\begin{equation}
\label{FOR-Probability}
    p = \textit{FOR}
\end{equation}

This probability is used to define the failure distribution of the software project. Hence, from Equations \ref{IndicatorRV} and \ref{FOR-Probability}, the probability distribution of the random variable $X_i$ is represented as:
\begin{equation*}
    \textbf{Pr}[X_i=1] = p, \text{ and, }  \textbf{Pr}[X_i=0]=1-p, \text{ for } 1 \leq i \leq l.
\end{equation*}

Now, to count the total failure instances from the prediction model, the following assumption ensures independence between each tested module:
\begin{assumption}
\label{Assumption4}
The SDP model provides predictions for independent observations (software modules).
\end{assumption}

In fact, all the SDP models assume independence between the data points \cite{khoshgoftaar1990predicting, lessmann2008benchmarking, herbold2017comparative, zhou2018far, zimmermann2009cross, amasaki2022extended, Sharma2022Measures}. Since each predicted clean module has a identical probability then it becomes a Bernoulli trial \cite{ross2014introduction}. Now, the sum of $l$ identical Bernoulli trials is said to be a binomial distribution \cite{ross2014introduction,motwani1995randomized}. This is given as: 
\begin{equation}
\label{SumofIndependentTrials}
    X = \sum_{i=1}^{l} X_i
\end{equation}

Now, the mean of the random variable $X$ is derived as follows (using linearity of expectation):
\begin{multline}
\label{Expectation}
    \mathbb{E}[X] = \mathbb{E}\Big[\sum_{i=1}^{l} X_i\Big] = \sum_{i=1}^{l} \mathbb{E}[X_i] 
    = \sum_{i=1}^{l} \big[1 * \textbf{Pr}[X_i=1] \\+ 0 * \textbf{Pr}[X_i=0]\big] = \sum_{i=1}^{l} p = lp
\end{multline}

So far, we have modelled the occurrence of failures (we also call it the hazard rate later in the paper) as a random variable and estimated the expected number of failures (wrong predictions for the defective modules) in a software. It is worth noting that, with no loss of generality, the predicted defective modules will be tested by the tester. The following assumption ensures the presence of failures in some portion of software after its release:
\begin{assumption}
\label{Assumption5}
For some portions of the software other than the predicted clean modules, the hazard rate follows the Weibull distribution.
\end{assumption}

Here, the hazard rate is defined as the instantaneous rate of failures in a software system \cite{lyu1996handbook}. According to Hartz et al. (in \cite{hartz1997introduction}), the hazards in a software (or the part of a software) may not be estimated with a single function. Hence, in order to fit various hazard curves, it is useful to investigate a hazard model of the form that is known as the Weibull distribution. Here, we assume the occurrence of a Weibull distribution of the hazard rate for the rest of the software modules (other than predicted clean modules). Now, for a software that is tested by using both the SDP model and the testers, the total estimated hazards ($\hat{z}(t)$) are calculated as:
\begin{equation}
\label{Modified Hazards}
    \hat{z}(t) = X + \hat{K}t^{\hat{m}}, \text{ for some } \hat{K}>0, \hat{m}>-1, \text{ and time } t>0
\end{equation}

Here, $\hat{z}(t)$ is the hazard rate of a software that is tested by using the SDP model (and later the predicted defective modules are serviced by the testers). Here, $\hat{K}t^{\hat{m}}$ is assumed to be the hazard rate, represented in terms of the Weibull distribution, for the software modules other than predicted clean modules. Here, the parameters $\hat{K}, \hat{m}$, and $t$ will take real values and the inequality constraints for these parameters are adopted directly from the Lyu's work \cite{lyu1996handbook}. Hence, from Assumptions \ref{Assumption1} and \ref{Assumption5}, for the total software modules, the resultant hazard model ($\hat{z}(t)$) is the sum of the hazard rates of the sub parts of the software. Now the expected hazard rate of the software is derived as:
\begin{equation}
\label{Expectation: Modified Hazards}
    \mathbb{E}[\hat{z}(t)] = \mathbb{E}[X + \hat{K}t^{\hat{m}}] = \mathbb{E}[X] + \hat{K}t^{\hat{m}} = lp + \hat{K}t^{\hat{m}}
\end{equation}

To demonstrate the feasibility of SDP models in the real-time scenario, the following assumptions must be met:
\begin{assumption}
\label{Assumption6}
An identical software is used for both the cases of testing using SDP model and manual testing.
\end{assumption}

This is an important assumption in providing proof for the feasibility of SDP models in the real-time scenario. For a software, from Equation \ref{Modified Hazards}, we know that the hazard rate is $\hat{z}(t) = X + \hat{K}t^{\hat{m}}$. Assume the same software that was tested by the testers, for which we have a Weibull distribution of the hazard rate as:
\begin{equation}
\label{Weibull-Hazard}
    z(t) = Kt^m, \text{ for some } K>0, m>-1, \text{ and time } t>0
\end{equation}

The definition of the parameters such as $K, m$ and $t$ is similar to the definition of the parameters in Equation \ref{Modified Hazards}. Note that, at time $t$, the two hazard functions such as $z(t)$ and $\hat{z}(t)$ describe the instantaneous rate of failures in the software when tested manually and with SDP, respectively. Now, for any software, the proofs (given in Section \ref{The Proofs}) provide the tight bounds for the deviation of a random variable far below from the corresponding hazard rate estimated with manual testing. Similarly, another possible approach is to find the deviation of the random variable (expressed in terms of reliability) far above the reliability of the manually tested software. Here, reliability is defined as the probability of failure-free software operation for a specified period of time in a specified environment \cite{lyu1996handbook}. 

The relation between reliability and hazard rate is given below \cite{lyu1996handbook}:
\begin{equation}
\label{Reliability from Hazard Rate}
    R(t) = e^{-\int_0^t z(x) dx}
\end{equation}

Where $R(t)$ is the software's reliability at time $t$, and $z(t)$ is the hazard rate. Using Equation \ref{Reliability from Hazard Rate}, we can derive the reliability values from numerous hazard models in some time interval [0,$t$]. Here, we assume that the two identical software systems (having different testing scenarios—one is manually testing and the other is testing using SDP) are deployed at time 0.

Now, The reliability of the manually tested software is now defined by the Weibull model of a hazard rate ($z(t)$).

\begin{lemma}
\label{Lemma-Reliability-Wibull}
For the Weibull hazard model of a software $z(t) = Kt^m, \text{ for some } K>0, m>-1 \text{ and time } t>0$, its reliability is:
\begin{align*}
    R(t) = e^{\frac{-Kt^{(m+1)}}{m+1}}
\end{align*}
\end{lemma}

\begin{proof}
Substituting the value of $z(t)$ (from Equation \ref{Weibull-Hazard}) in  $R(t)$ (that is, in Equation \ref{Reliability from Hazard Rate}) yields:
\begin{align}
\label{Reliability-Weibull-Inintial}
    R(t) = e^{-\int_0^t Kx^m \text{ } dx}
\end{align}
Simplifying Equation \ref{Reliability-Weibull-Inintial} satisfies the Lemma.
\end{proof}

The proofs in Section \ref{The Proofs} are valid by ensuring the probability value ($p$) lies in the interval (0,1). Hence, the following assumption must hold true:
\begin{assumption}
\label{Assumption7}
The SDP model should produce at least one false negative and one true negative on the test set.
\end{assumption}

\section{The Proofs}
\label{The Proofs}
\subsection{The tight lower bound in terms of hazard rate}
\label{Bound on the Hazard Rate}
In Section \ref{Preliminaries}, we modelled the number of hazard (failure) instances in a software that is tested by using the SDP model as a random variable (that is, $\hat{z}(t) = X + \hat{K}t^{\hat{m}}$). Now, the following theorem defines the deviation of a random variable, $X + \hat{K}t^{\hat{m}}$ below the value of hazard rate of a manually tested software, $Kt^m$ (in fact, far below from the expectation, $\mu$). 
\begin{theorem}
\label{Theorem-Weibull-Hazard}
Let $X_1, X_2, \dots, X_l$ be the independent Bernoulli trials such that for, $1\leq i \leq l$, $Pr[X_i = 1] = p$, where $0 < p < 1$. Also let $\exists$ parameters $K>0, m>-1, \hat{K}>0, \hat{m}>-1$, time $t>0$. Then for X = $\sum_{i=1}^{l} X_i, \hat{z}(t) = X + \hat{K}t^{\hat{m}}, \mu = \mathbb{E}[X + \hat{K}t^{\hat{m}}] = \hat{K}t^{\hat{m}} + lp$, and for the Weibull hazard model of a manually tested software, $Kt^m$:
\begin{align*}
    Pr[X + \hat{K}t^{\hat{m}} < Kt^m] < e^{\frac{-(lp-Kt^m+2\hat{K}t^{\hat{m}})^2}{2[\hat{K}t^{\hat{m}} + lp]}}
\end{align*}
\end{theorem}
\begin{proof}
As before, $Pr[X + \hat{K}t^{\hat{m}} < Kt^m]$ can be rewritten as:
\begin{equation}
\label{Rewrite the hazard inequality}
    Pr[X + \hat{K}t^{\hat{m}} < Kt^m] = Pr[X < Kt^m - \hat{K}t^{\hat{m}}]
\end{equation}

We know for some $0<\delta \leq 1$, and $\mu$,  using the Chernoff bound, the lower tail bound for the sum of independent Bernoulli trials, $X$, that deviates far from the expectation $\mu$ is \cite{chernoff1952measure}:
\begin{align}
\label{LowerChernoff}
    Pr[X < (1-\delta)\mu] < e^{\frac{-\mu\delta^2}{2}}
\end{align}

Here, the value $(1-\delta)\mu$ represents the left-side marginal value from the expectation $\mu$ with the band length of $\delta$.

Now, we wish to obtain a tight lower bound that the random variable (that is, $X + \hat{K}t^{\hat{m}}$), that deviates far below from the hazard rate of a manually tested software, $Kt^m$. In Equation \ref{Rewrite the hazard inequality}, for some $K>0, \hat{K}>0, m>-1, \hat{m}>-1$, and $t>0$, the value $Kt^m - \hat{K}t^{\hat{m}}$ is assumed to be below the expectation, $\mu$, in a given time period $[0,t]$. Now, from Equations \ref{Rewrite the hazard inequality} and \ref{LowerChernoff}:
\begin{align*}
    (1-\delta)\mu = Kt^m - \hat{K}t^{\hat{m}}
\end{align*}
\begin{align}
\label{Weibull-Hazard-Delta}
    \Rightarrow \delta = 1 - \frac{Kt^m-\hat{K}t^{\hat{m}}}{\mu}
\end{align}

From Equation \ref{Modified Hazards}, we know the expected hazards in a software, which uses the SDP models is:
\begin{align*}
    \mu = \mathbb{E}[X + \hat{K}t^{\hat{m}}] = \hat{K}t^{\hat{m}} + lp
\end{align*}

Now, substituting $\delta$ (from Equation \ref{Weibull-Hazard-Delta}) and $\mu$ in Equation \ref{LowerChernoff} results in the tight lower bound for the deviation of a random variable (that is, $X + \hat{K}t^{\hat{m}}$) below the hazard rate of a manually tested software. This is expressed below:
\begin{align}
     Pr[X + \hat{K}t^{\hat{m}} < Kt^m] < e^{\frac{-[\hat{K}t^{\hat{m}} + lp]\big[1 - \frac{Kt^m-\hat{K}t^{\hat{m}}}{\hat{K}t^{\hat{m}} + lp}\big]^2}{2}}
\end{align}
Simplifying the above equation will ensure the proof.
\end{proof}

Thus, we have from Theorem \ref{Theorem-Weibull-Hazard} the occurrence of fewer hazards in the software that uses SDP than the occurrence of the total hazards in the same software that is tested by a human is exponentially small in $l, \hat{K}, \hat{m}$, and $t$, implying that at the larger values of these parameters, the bound becomes tighter.
\subsection{The tight upper bound in terms of reliability}
\label{Bound on the Reliability}
In this section, we provide a lemma that calculates the reliability of a software that is tested using the SDP model.
\begin{lemma}
\label{Lemma-Reliability-X}
Let $X_1, X_2, \dots, X_l$ be the independent Bernoulli trials, also let $\exists$ parameters $\hat{K}>0, \hat{m}>-1$, time $t>0$. Then for X = $\sum_{i=1}^{l} X_i$ and $\hat{z}(t) = X + \hat{K}t^{\hat{m}}$, its reliability is:
\begin{align*}
    \hat{R}(t) = e^{-\Big[Xt+\frac{\hat{K}t^{\hat{m}+1}}{\hat{m}+1}\Big]}
\end{align*}
\end{lemma}
\begin{proof}
From Equation \ref{Modified Hazards}, we have the hazards in software, which is tested by using the SDP model. Now substitute the value of $\hat{z}(t)$ (from Equation \ref{Modified Hazards}) in Equation \ref{Reliability from Hazard Rate}, then we have:
\begin{align}
\label{Reliability-X-Initial}
    \hat{R}(t) = e^{-\int_0^t  [X + \hat{K}x^{\hat{m}}] dx}
\end{align}

Here, $\hat{R}(t)$ is a random variable used to represent the reliability of the software which is tested from the predictions of the SDP model. Now, simplifying Equation \ref{Reliability-X-Initial} will result in the reliability of the software that is tested by using the SDP model.
\end{proof}
Now, the expected reliability of a software (which uses SDP models), $\mu_{\hat{R}}$ or $\mathbb{E}[\hat{R}(t)]$ is derived as:
\begin{multline}
\label{Expectation: Modified Reliability-1}
    \mathbb{E}[\hat{R}(t)] = \mathbb{E}\Big[e^{-\Big[Xt+\frac{\hat{K}t^{\hat{m}+1}}{\hat{m}+1}\Big]}\Big] = \mathbb{E}\big[e^{-Xt}\big]e^{\frac{\hat{K}t^{\hat{m}+1}}{\hat{m}+1}}
\end{multline}
We observe that:
\begin{equation}
    \mathbb{E}\big[e^{-Xt}\big] = \mathbb{E}\big[e^{-t\sum_{i=1}^{l}X_i}\big] = \mathbb{E}\big[\prod_{i=1}^{l} e^{-tX_i}\big]
\end{equation}

Since the $X_i$ are independent, the random variables $e^{-tX_i}$ are also independent. It follows that, $\mathbb{E}\big[\prod_{i=1}^{l} e^{-tX_i}\big]=\prod_{i=1}^{l} \mathbb{E}\big[e^{-tX_i}\big]$. Now using these facts in Equation \ref{Expectation: Modified Reliability-1} gives:
\begin{equation}
\label{Expectation: Modified Reliability-2}
    \mathbb{E}[\hat{R}(t)] = e^{\frac{\hat{K}t^{\hat{m}+1}}{\hat{m}+1}} \prod_{i=1}^{l} \mathbb{E}\big[e^{-tX_i}\big]
\end{equation}

Here, the random variable $e^{-tX_i}$ assumes a value $e^{-t}$ with probability $p$, and the value 1 with probability $1-p$. Now computing $\mathbb{E}\big[e^{-tX_i}\big]$ from these values, we have that:
\begin{equation}
\label{Expectation: Modified Reliability-3}
    \prod_{i=1}^{l} \mathbb{E}\big[e^{-tX_i}\big] = \prod_{i=1}^{l}\big[pe^{-t}+1-p\big] = \prod_{i=1}^{l} \big[1+p(e^{-t}-1)\big]
\end{equation}

Now we use the inequality $1+x<e^x$ with $x = p(e^{-t}-1)$, to obtain the expected reliability.
\begin{equation}
\label{Expectation: Modified Reliability}
    \mu_{\hat{R}} = \mathbb{E}[\hat{R}(t)] < e^{\Big[lp(e^{-t}-1)+\frac{\hat{K}t^{\hat{m}+1}}{\hat{m}+1}\Big]}
\end{equation}

The intuition behind the inequality is to provide an easy computation and that does not harm the final bound in the following theorem. Now, by using the Lemmas \ref{Lemma-Reliability-Wibull} and \ref{Lemma-Reliability-X}, the following theorem provides a bound for the deviation of a random variable, $e^{-\Big[Xt+\frac{\hat{K}t^{\hat{m}+1}}{\hat{m}+1}\Big]}$, above the reliability of a manually tested software, $e^{\frac{-Kt^{(m+1)}}{m+1}}$. 

\begin{theorem}
\label{Theorem-Weibull-Reliability}
Let $X_1, X_2, \dots, X_l$ be the independent Bernoulli trials such that for, $1\leq i \leq l$, $Pr[X_i = 1] = p$, where $0 < p < 1$. Also let $\exists$ parameters $K>0, m>-1, \hat{K}>0, \hat{m}>-1$, time $t>0$. Then for X = $\sum_{i=1}^{l} X_i, \hat{R}(t) = e^{-\Big[Xt+\frac{\hat{K}t^{\hat{m}+1}}{\hat{m}+1}\Big]}, \mathbb{E}[\hat{R}(t)] = \mu_{\hat{R}} < e^{\Big[lp(e^{-t}-1)+\frac{\hat{K}t^{\hat{m}+1}}{\hat{m}+1}\Big]}$, and for the Weibull distribution for the Reliability function, $e^{\frac{-Kt^{m+1}}{m+1}}$:
\begin{multline*}
    Pr\Big[e^{-\Big[Xt+\frac{\hat{K}t^{\hat{m}+1}}{\hat{m}+1}\Big]} > e^{\frac{-Kt^{m+1}}{m+1}}\Big] <\\ e^{-\Bigg[\Big[e^{\Big[lp(e^{-t}-1)+\frac{\hat{K}t^{\hat{m}+1}}{\hat{m}+1}\Big]}-\big[\frac{Kt^m}{m+1}-\frac{\hat{K}t^{\hat{m}}}{\hat{m}+1}\big]\Big]^2\frac{1}{2e^{\Big[lp(e^{-t}-1)+\frac{\hat{K}t^{\hat{m}+1}}{\hat{m}+1}\Big]}}\Bigg]}
\end{multline*}
\end{theorem}
\begin{proof}

The proof for this upper tail is very similar to the proof for the lower tail, as we saw in Theorem \ref{Theorem-Weibull-Hazard}. As before,
\begin{multline}
\label{Rewrite the Reliability inequality}
     Pr\Big[e^{-\Big[Xt+\frac{\hat{K}t^{\hat{m}+1}}{\hat{m}+1}\Big]} > e^{\frac{-Kt^{m+1}}{m+1}}\Big] =
     Pr\Big[Xt < \\ \Big[\frac{Kt^{m+1}}{m+1}-\frac{\hat{K}t^{\hat{m}+1}}{\hat{m}+1}\Big]\Big]
     = Pr\Big[X < \Big[\frac{Kt^m}{m+1}-\frac{\hat{K}t^{\hat{m}}}{\hat{m}+1}\Big]\Big]
\end{multline}

Now, we wish to obtain a tight lower bound that the random variable, $e^{-\Big[Xt+\frac{\hat{K}t^{\hat{m}+1}}{\hat{m}+1}\Big]}$, deviates far from the value $e^{\frac{-Kt^{m+1}}{m+1}}$. In Equation \ref{Rewrite the Reliability inequality}, for some $K>0, \hat{K}>0, m>-1, \hat{m}>-1$, and $t>0$, the value $\Big[\frac{Kt^m}{m+1}-\frac{\hat{K}t^{\hat{m}}}{\hat{m}+1}\Big]$ is assumed to be below the expectation, $\mu_{\hat{R}} = \mathbb{E}[\hat{R}(t)]$, in a given time period $[0,t]$. Now, equating the Equations \ref{LowerChernoff} and Equation \ref{Rewrite the Reliability inequality}, then we get:
\begin{align*}
    (1-\delta)\mu_{\hat{R}} = \frac{Kt^m}{m+1}-\frac{\hat{K}t^{\hat{m}}}{\hat{m}+1}
\end{align*}
\begin{align}
\label{Weibull-Reliability-Delta}
    \Rightarrow \delta = 1 - \Big[\frac{Kt^m}{m+1}-\frac{\hat{K}t^{\hat{m}}}{\hat{m}+1}\Big]\frac{1}{\mu_{\hat{R}}}
\end{align}

Now, substitute the value of $\delta$ (from Equation \ref{Weibull-Reliability-Delta}) in Equation \ref{LowerChernoff} to obtain the tight upper bound (expressed in terms of lower bound) for the deviation of a random variable (that is, $e^{-\Big[Xt+\frac{\hat{K}t^{\hat{m}+1}}{\hat{m}+1}\Big]}$) from the reliability (which is derived from the Weibull model of the hazard rate) of a manually tested software $e^{\frac{-Kt^{m+1}}{m+1}}$. This is expressed below:
\begin{align}
    Pr\Big[e^{-\Big[Xt+\frac{\hat{K}t^{\hat{m}+1}}{\hat{m}+1}\Big]} > e^{\frac{-Kt^{m+1}}{m+1}}\Big] < e^{\frac{-\mu_{\hat{R}}\Bigg[1-\frac{1}{\mu_{\hat{R}}}\Big[\frac{Kt^m}{m+1}-\frac{\hat{K}t^{\hat{m}}}{\hat{m}+1}\Big]\Bigg]^2}{2}}
\end{align}
After simplification, we get: 
\begin{align}
\label{Final Bound using Reliability}
    Pr\Big[e^{-\Big[Xt+\frac{\hat{K}t^{\hat{m}+1}}{\hat{m}+1}\Big]} > e^{\frac{-Kt^{m+1}}{m+1}}\Big] < e^{-\Big[\mu_{\hat{R}}-\big[\frac{Kt^m}{m+1}-\frac{\hat{K}t^{\hat{m}}}{\hat{m}+1}\big]\Big]^2\frac{1}{2\mu_{\hat{R}}}}
\end{align}

Substituting the expected reliability $\mu_{\hat{R}}$ from Equation \ref{Expectation: Modified Reliability} in Equation \ref{Final Bound using Reliability} accomplishes the proof.
\end{proof}
Thus, we have from Theorem \ref{Theorem-Weibull-Reliability}, the possibility of getting better reliability in the software that uses SDP than in the same software that is tested by a human is exponentially small in $l, \hat{K}, \hat{m}$, and $t$, implying that, similar to the result of Theorem \ref{Theorem-Weibull-Hazard}, at the larger values of these parameters, the bound becomes tighter.

\section{Future Plans}
\label{Future Plans}
Theorems \ref{Theorem-Weibull-Hazard} and \ref{Theorem-Weibull-Reliability} provides preliminary bounds for the post-analysis of the binary classification model (SDP model) in real-time working environments. We believe that providing a critique of the developed binary classification model in the real-time working environment is novel in machine learning theory and has the potential to provide insight into the feasibility of other applications (such as safety-critical applications, for example, tumour prediction systems for medical diagnosis, online fraud detection, etc.). Within the scope of this work, the extensions of Theorems \ref{Theorem-Weibull-Hazard} and \ref{Theorem-Weibull-Reliability} are numerous. A few examples include, 1) the bounds become more specific to the application if the state-of-the-art hazard (and reliability) models are used in the construction of the proof, 2) new bounds derived if the random variable $X$ is assumed to be a function of time, $t$, and 3) new bounds derived assuming the dependency among the random variables (relaxing Assumption \ref{Assumption4}). In this case, we derive bounds assuming the presence of cascading failures in the software as a result of SDP model.  

\bibliographystyle{IEEEtran}
\bibliography{ICSE-NIER}

\end{document}